\documentclass[conference]{IEEEtran}
\IEEEoverridecommandlockouts
\usepackage{cite}
\usepackage{amsmath,amssymb,amsfonts}
\usepackage{textcomp}
\usepackage{xcolor}
\usepackage[utf8]{inputenc}
\usepackage{amsmath,amssymb}
\usepackage{float}
\usepackage{epsfig}
\usepackage{graphics}
\usepackage{graphicx}
\usepackage{color}
\usepackage{latexsym}
\usepackage{epstopdf}
\usepackage{subcaption}
\usepackage{algorithm}
\usepackage{algpseudocode}
\usepackage{pifont}
\usepackage{bbm}
\usepackage{textcomp}
\usepackage{amsthm}

\usepackage{url}

\newtheorem{definition}{Definition}

\DeclareMathOperator*{\argmax}{arg\,max}

\IEEEoverridecommandlockouts
\newtheorem{theorem}{Theorem}[]
\newtheorem{lemma}[theorem]{Lemma}

\usepackage{authblk}

\title{Strategy-Proof Spectrum Allocation among Multiple Operators}

\author[1]{Indu Yadav}
\author[2]{Ankur A. Kulkarni}
\author[1,3]{Abhay Karandikar}

\affil[1]{Department of Electrical Engineering, Indian Institute Technology Bombay, India 400076\authorcr Email: {$\lbrace$indu, karandi$\rbrace$@ee.iitb.ac.in}}
\affil[2]{Systems and Control Engineering, Indian Institute Technology Bombay, India 400076\authorcr Email: {kulkarni.ankur@iitb.ac.in}}
\affil[3]{Director and Professor, Indian Institute Technology Kanpur, India 208016\authorcr Email: {karandi@iitk.ac.in}}

\begin{document}


\maketitle

\begin{abstract}
To address the demand of exponentially increasing end users efficient use of limited spectrum is a necessity. For this, spectrum allocation among co-existing operators in licensed and unlicensed spectrum band is required to cater to the temporal and spatial traffic variations in the wireless network. 
In this paper, we consider multiple operator spectrum allocation problem via auctions. The classical Vickrey-Clarke-Grooves (VCG) approach provides a strategy-proof and social welfare maximizing auction at the cost of high computational complexity which makes it intractable for practical implementation. We propose a sealed bid auction for spectrum allocation, which is computationally tractable and can hence be applied as per the dynamic load variations of the network. We  show that the proposed algorithm is strategy-proof. Simulation results are presented to exhibit the performance comparison of the proposed algorithm and the VCG mechanism.

\end{abstract}


\section{Introduction}
The telecom market has been growing steadily since the past few decades. Based on studies by the Ericsson Mobility Report, 98 million new mobile subscriptions were reported globally in the first quarter of 2017~\cite{barboutov2017ericsson}. It is also predicted that the number of cellular Internet of Things (IoT) connections is expected to reach $3.5$ billion by $2023$ with Long Term Evolution (LTE) constituting major chunk of the market. To fulfill the requirements of the exponentially increasing end users, additional spectrum is needed. This also requires more base station deployment in the region. However, it is well known that the wireless spectrum is a scarce and limited resource.


 \par Traditionally, the spectrum is allocated 
for large period of time (one or more years) to the service providers using auctions.
 The operators demand spectrum as per the peak time requirement of the network. This leads to inefficient usage of the scarce spectrum resources as it has been observed that most of the time full spectrum chunk is not in use. To overcome the inefficient spectrum usage in LTE cellular bands or below 1GHz, allocation should be done as per the spatial and temporal traffic conditions of the network. For this, a computationally efficient spectrum allocation mechanism is required which can be implemented on the fly for dynamic load variations in the network. In our work, we focus on spectrum allocation mechanism among multiple operators via a computationally efficient auction.

\par Spectrum allocation among the service providers is commonly performed using a sealed bid auction format. In sealed bid auctions, each service provider sends the demand along with the price valuation in a closed envelope to the auctioneer. This ensures that the information of its demand and the private valuation is not known to other service providers or participants of the auction. However, in auctions, strategy-proofness \cite{krishna2009auction} is a major concern. Strategy-proofness implies that the participants in the auction should not gain by deviating from its true valuation. In other words, one should not gain by mis-reporting about one's valuation or demand. In general, the objective behind the spectrum auctions is social-welfare maximization instead of revenue maximization. Social-welfare maximization ensures that spectrum is allocated to those who value it most due to limited availability. Ideally, in spectrum auction mechanisms one may want an algorithm that satisfies three properties: strategy-proofness, social-welfare maximization  and computational feasibility. Vickrey-Clarke-Grooves (VCG) \cite{vickrey1961counterspeculation,clarke1971multipart,groves1973incentives} is one algorithm for spectrum allocation that satisfies strategy-proofness and provides optimal social welfare but it is not computationally efficient. 
\par The reason behind the emphasis on the strategy-proof property in the auctions is to make the process of resource allocation easier for the auctioneer as well as for the participants. It removes the time and computational overhead involved in strategy making for the optimal bid value determination from the user's perspective as bidding at true valuation is always the dominant strategy. This ensures bidding at true valuation will always lead to the maximum utility gain. This encourages more participants in the auction. From the auctioneer's perspective more participants and true bidding would help to achieve better revenue in the auction.  This is also beneficial in terms of conducting the repeated auctions of short duration, with the elimination of time overhead caused to strategize the mechanism to manipulate the market to gain the utility with respect to other bidders in the auction.  

\par Current approaches \cite{zhou2008ebay} \cite{wu2011small} in spectrum  allocation assume one base station per operator. However, in any real scenario multiple operators provide services using multiple base stations deployed (distributed) across the region.  Here, spectrum valuation varies across the base stations  individually for an operator depending on traffic conditions of the location of the base station. Therefore, valuation of an operator in the region depends on  the set of base stations of the operator who are allocated channels. Thus, the spectrum allocation among the base stations of the multiple operators becomes a multi-parameter environment problem. 
\par In our work, we propose an algorithm which achieves strategy-proofness and is computationally feasible at the cost of optimal social welfare. Thanks to its computational efficiency, the algorithm can also be used for real time spectrum allocation. Unlike previous approaches, our algorithm allows multiple operators with multiple base stations.

\subsection{Related Work}
In \cite{zhou2008ebay} authors propose a sealed bid auction strategy-proof mechanism which follows monotonicity, i.e, if a base station is allocated channel at bid value $b$, then it would also get channel at any bid value $b' \geq b$, provided bid values of other base stations are kept constant. SMALL \cite{wu2011small} is another mechanism which shows improvement in allocation efficiency than the algorithm proposed in \cite{zhou2008ebay}, though some winners are \textquoteleft sacrificed\textquoteright to ensure the strategy-proofness of algorithm. In \cite{gandhi2008towards} authors propose an approach for fine grained allocation of channel (where a channel is further sliced into smaller frequencies), but it does not satisfy the strategy-proof criteria. To address the interference issue in  spectrum allocation,  auction based power allocation mechanism has been proposed in \cite{clemens2005intelligent}, which fail to be incentive compatible. In \cite{subramanian2008near}, authors proposed
algorithm based on greedy graph coloring approach for spectrum allocation that maximizes the revenue of the auction.

\par The algorithms proposed in \cite{zhou2008ebay,wu2011small} are restricted to one base station per operator, i.e., each base station individually acts as a player in the auction whereas in our work we have considered a multi-parameter environment among the multiple base stations of the operators. In our work, non co-operative behavior of operators is also considered which holds as per the practical scenarios.

\subsection{Contributions}
Our contributions in the work are as follows:
\begin{itemize}
\item We study the spectrum allocation among multiple base stations of multiple operators in a multi-parameter environment, which has not been addressed in the literature so far.
\item We propose a strategy-proof spectrum allocation mechanism in the multi-parameter environment based on the sealed bid auction format where strategy-proofness holds for a vector of bids corresponding to an operator. 
\item We perform a comparison of computational complexity between the proposed algorithm and the VCG algorithm and observe that the proposed algorithm is tractable for large number of base stations. Thus, the algorithm is practically feasible for real time allocation of spectrum.
\end{itemize}

\par The rest of the paper is organized as follows. Section \ref{sec:sys_model} describes the system model along with the problem formulation. In Section \ref{sec:mec_design}, we propose the mechanism design. In Section \ref{sec:comp_complexity}, computational complexity of the proposed algorithm is investigated in comparison to VCG. In Section \ref{sec:sim_results}, we present the simulation results. In Section \ref{sec:conclude}, we conclude the paper.

\section[\textwidth=16 cm]{System Model}
\label{sec:sys_model}
In this section, we discuss the system model for spectrum allocation among the multiple operators which provide services in a geographical region.  System model comprises of a spectrum database, auctioneer, controllers and base stations. The spectrum database contains information about the spectrum chunk available for allocation in the region. The spectrum chunk is further divided into orthogonal channels (non-overlapping frequency band) for auction. Auctioneer is responsible for channel allocation among the base stations of the operators and also takes decision about the prices to be charged from the operator for providing right to use a channel at the particular base station. The role of auctioneer is to come up with a strategy for allocation of the channel. Each operator owns a controller, which is responsible for the calculation of demand (number of channels) as well as the valuation (price willing to pay in case allocated) at each base station of the operator individually and communicates these to the auctioneer.

\begin{figure}[h]
\centering
\includegraphics[width = 0.40\textwidth]{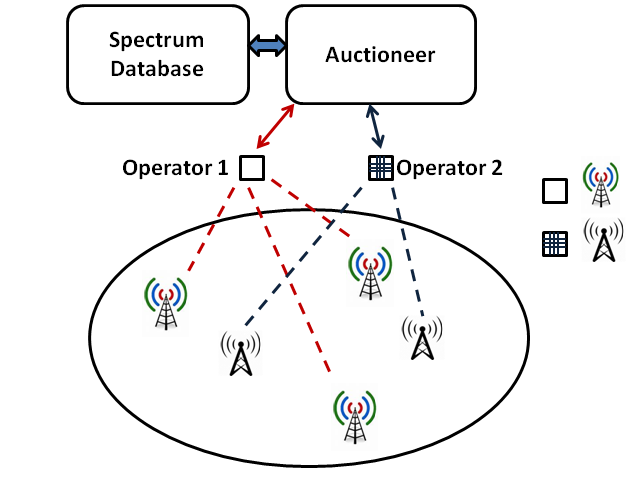}
\caption{System model}
\label{sys_model}
\end{figure}
\par Spectrum allocation model among the base stations of two operators in a given region is illustrated in Figure \ref{sys_model}. The operators $1$ and $2$ having $2$ and $3$ base stations, respectively are deployed.  Unlike the existing work in which individual base stations act as bidders, in our system model an operator acts as a bidder for the set of base stations associated with itself. Since an operator acts as an agent for its base stations, it communicates a vector of bids (valuation) to the auctioneer via the controller. The valuation of a channel at each base station may be different. Valuation at each base station indicates the price that an operator is willing to pay in return of getting the \textquoteleft right to use\textquoteright a channel at the particular base station. Therefore, for different sets of base stations an operator has different valuations, which qualifies the problem of spectrum allocation as  multi-parameter environment problem \cite{roughgarden2016twenty}. 

\par We make the following assumptions:

\begin{enumerate}
 \item We assume that an auctioneer has the knowledge of the topology in the geographical region. Therefore, the overall conflict graph consisting of all the base stations participating in the auction is available to the auctioneer.
 \item We assume all channels are homogeneous in characteristics and act as substitutes. Thus, the bid valuation is channel independent. 
 \item For simplicity, we assume that the base stations that belong to the same operator do not conflict with each other. Hence, any base station of an operator would experience interference only from the base stations of other operators in the region. 
\end{enumerate}
\par We capture the interference among the base stations of the operators with the help of a graph $\mathcal{G} = (V,\mathcal{E})$, that is obtained from the knowledge of topology in the geographical region, where $V$ represents the set of vertices (nodes), and $\mathcal{E}$ represents the set of edges in the graph. The set of vertices in the graph correspond to the base stations of various operators in the region. Any two base stations are said to interfere with each other, if the geographical distance between them is less than a predetermined value $d$. In this case, there is an edge between them in the graph. Two interfering base stations (nodes) can not be assigned same channel concurrently. 

\subsection{Strategy-Proof Auctions}
\label{problem_form}
 
 Spectrum auctions are different from the conventional auctions as the same channel can be reallocated after certain distance depending on the interference criteria. This further introduces challenges in the spectrum auctions in achieving strategy-proof auction. In conventional auctions (where an object can be provided to only one bidder), second price auction mechanism ensures strategy-proof auction. In a second price auction, bidders are arranged in the decreasing order of their bids, and the highest bidder gets the resources at the valuation of the second highest bidder in the list. However, in spectrum auctions, the second price auction mechanism no longer remains truthful \cite{zhou2008ebay}. Strategy-proof spectrum allocation among the multiple operators with multiple base stations becomes more challenging in the auctions.  All base stations associated with an operator may not be causing interference to all the base stations of the other operators in the region, then charging the second highest valuation from the highest bidding operator in the list of operator valuations does not consider the interference criteria.  
  
   VCG is another mechanism which always remains truthful by selecting the optimal outcome i.e., it chooses the winner set that maximizes the valuation of the auction. But, the pricing mechanism involves high computational complexity to achieve the strategy-proofness in the auction. This makes the VCG mechanism intractable for channel allocation in spectrum auctions with large set of base stations.

In the Section \ref{sec:mec_design}, we propose a computationally efficient and strategy-proof auction mechanism for spectrum allocation among the multiple operators. The proposed algorithm assumes only one channel is available for auction across base stations of the multiple operators providing service in the region. This mechanism can be further extended in case of multiple channel availability.

\subsection{Notations and Definitions }
 We introduce the following notations:
\begin{itemize}
 \item $N = \{1,2,\ldots,n\}$ represents the set of operators participating in the spectrum auction in the geographical region.
 \item $m_{i}$ represents the number of base stations corresponding to an operator $i$.
 \item $\mathcal{S}_i = \{S_{1}^{i},S_{2}^{i},\ldots,S_{m_i}^{i}\}$ represents the set of base stations of an operator $i$.
 \item $v_{ij}$ represents the true valuation of $S_{j}^{i}$ i.e, $j^{th}$  base station corresponding to an operator $i$ .
 \item $\vec{v}_i = \{v_{i1},v_{i2},\ldots\,v_{im_i}\}$ represents the vector of true valuation at base stations for operator $i$.
 \item $b_{ij}$ represents the bid of $S_{j}^{i}$ i.e, $j^{th}$  base station corresponding to an operator $i$.
 \item $\vec{b}_i = \{b_{i1},b_{i2},\ldots\,b_{im_i}\}$ represents the vector of bids for operator $i$.
 \item $\mathcal{N}_i$ represents the set of neighboring base stations which are in conflict with the base stations of operator $i$, such that the same channel cannot be shared simultaneously.
 \item $X_{i}$ represents the binary allocation vector corresponding to operator $i$, where $1$ represents channel allocation.
 \item $X_{i}(j)$ represents the $j^{th}$ component of $X_i$ vector.
 \item $O_i$ represents operators that are in neighbors of $i$ i.e.  ($ \{ operators ~j ~|~ S_j \bigcap \mathcal{N}_i \neq \phi, j \neq i\})$.
 \item $d_i = \{d_{i1},d_{i2},\ldots,d_{im_i}\}$ represents the number of channels required at base stations of operator $i$.
 \item $N(\mathcal{G}^{'})$ represents the set of active operators from the conflict graph $\mathcal{G}^{'}$.
\end{itemize}
\begin{enumerate}
 \item \textit{True valuation} ($\sigma_i^{v}$) : True valuation $\sigma_i^{v}$ of any operator $i$ is defined as the sum of the actual valuations (actual valuations are private and not known to the auctioneer) of all the base stations corresponding to the operator $i$.
 \begin{equation}
 \label{eqn:eqn1}
 \sigma_i^{v} = \sum_{j=1}^{m_i} v_{ij}.
 \end{equation}
 \item \textit{Bidding valuation} ($\sigma_i^{b}$) : Bidding valuation $\sigma_i^{b}$ of any operator $i$ is defined as the sum of the bids (which may or may not be same as the actual valuation) of all the base stations corresponding to the operator $i$.
  \begin{equation}
  \label{eqn:eqn2}
 \sigma_i^{b} = \sum_{j=1}^{m_i} b_{ij}.
 \end{equation}
 \item \textit{Price} ($p_{i}$): It is defined as the price that an operator $i$ has to pay, in case operator $i$ wins the resources (channels), else it would be zero.

 \item \textit{Operator Utility $(\mathcal{U}_i)$} : Utility of an operator $i$ is the difference between the operator valuation (unknown to the auctioneer) and the price charged on the allocation of the channel. If the operator does not get the channel, the utility would be zero. In other words, it gives the overall gain of an operator if it is allocated channels.
 \begin{equation}
 \label{eqn:eqn4}
 \mathcal{U}_i(b_i,b_{-i}) =  \begin{cases} \sigma_i^{v} - p_{i}& ~\text{if the channel is allocated} \\ 0 & ~\text{otherwise}. \end{cases}.
 \end{equation}
\end{enumerate}
 where, $b_i$ is the bid vector of operator $i$ and $b_{-i}$ bid vectors of other operators except $i$.
\begin{definition} \label{def: strategy-proofness}
 An auction is truthful (strategy-proof) if there is no incentive in  deviating  from  the  true  valuation. Thus,  the dominant  strategy  is  to  bid  at the true valuation no matter what strategy others are choosing.
\end{definition}
\begin{equation}
 \label{eqn:eqn5}
 \mathcal{U}_{i}(b_i, b_{-i}) \leq \mathcal{U}_{i}(v_i, b_{-i}) \quad \forall b_{i}, \& b_{-i}.
 \end{equation}
 where, $v_i$ is the vector of true valuations at the base stations of an operator $i$.

\section{Mechanism Design}
\label{sec:mec_design}
In this section, we describe the mechanism design for channel allocation among the base stations of multiple operators. Recall the assumptions made in previous sections, multiple base stations correspond to an operator and one channel availability. In auctions, the mechanism design has two steps: channel allocation and price charging strategy. In channel allocation auctioneer decides whom the right to use the channel is provided, and what price should be charged is decided in pricing strategy. The price charged enforces the operators to declare the true valuation to ensure strategy-proof auction. Now, we define critical operator which is used later in the price charging strategy.
\begin{definition} \label{def: critical neighbor}
A critical operator $C(i)$ of an operator $i$ is defined as the operator in $\mathcal{N}_i$, whose sum of bids is maximum among all the operators in $\mathcal{N}_i$ except $i$.
\end{definition}
\begin{equation}
 \label{eqn:eqn6}
 \begin{split}
 C(i) = \Big\{ j \in O_i | \sum_{k \in \{ \mathcal{N}_i\bigcap S_j\}} b_{jk} \geq \sum_{k \in \{ \mathcal{N}_i\bigcap S_j'\}}b_{j'k},\\
 \quad \forall j' \neq j,~i  ~\& j' \in O_i\Big\}.
 \end{split}
 \end{equation}
 Let us define a set $\mathcal{L}_{k}^{i}= \mathcal{N}_i \cap S_k$, which contains the base stations of operator $k$ in conflict with the base stations of the operator $i$. Let $\Lambda_k^{i}$ be the valuation of set $\mathcal{L}_{k}^{i}$ which is given as, $\Lambda_k^{i} = \sum b_{kj}\mathbbm{1}_{\{S_{kj} \in \mathcal{L}_{k}^{i} \}}$. The critical operator of an operator $i$ can be obtained as, $C(i) = \argmax \limits_{k \neq i}\Lambda_k^{i}, ~k \in O_i$ and the critical operator valuation $\sigma_i^{c}$ is given as, $\sigma_i^{c} = \max \limits_{k\neq i} \Lambda_k^{i},k \in O_i$.
 
 
 The strategy-proof algorithm proposed is described in Algorithm \ref{channel_allocation}. 
\begin{algorithm}
\caption{Strategy-proof auction mechanism}
\label{channel_allocation}
\begin{algorithmic}[1]

\State \textbf{Input: }Conflict Graph $\mathcal{G}$, bid vector, $ \{ \vec{b}_i \}_{\{ i \in N \}}$.
 \State\textbf{Output: }Binary channel allocation vector $\{ X_i\}_{\{ i \in N \}}$, payment $\{p_i\}_{\{ i \in N \}}$.
 \State Initialize vector $X_i \leftarrow \vec{0}$, $N(\mathcal{G})= \{ 1,2,\ldots, N\}$
\State Initialize $p_i \leftarrow 0$, $\mathcal{G^{'}} \leftarrow \mathcal{G}$, $N(\mathcal{G}^{'})= N(\mathcal{G})$, $FLAG = True$.
 \While  {$(FLAG = True)$}
   \State Calculate $i^* = \argmax \limits_{ \{ i \in N(\mathcal{G}^{'})\}}\sigma_i^{b}$.
   \State Find $\mathcal{N}_{i^*}$.
   \State Calculate $C(i^*) = \argmax \limits_{k \neq i^*}\Lambda_k^{i^*},~k \in O_i*$ and $\sigma_{i^{*}}^{c} = \max \limits_{k\neq i^*} \Lambda_k^{i^*}~k \in O_i*$.
   \State Make $p_{{i}^*} \leftarrow \sigma_{i^{*}}^{c}$ and Make $X_{{i}^*} \leftarrow 1$.
     \If {$ (\mathcal{G}^{'} \cap (S_{i^*} \cup \mathcal{N}_{i^*}) = \mathcal{G}^{'})$}
       \State  $ FLAG \leftarrow False. $
     \Else 
       \State $\mathcal{G}^{'} \leftarrow \mathcal{G}^{'} \backslash \{ S_{i^*} \cup \mathcal{N}_{i^*} \}$.
     \EndIf 

\EndWhile
\end{algorithmic}
\end{algorithm}
  This algorithm takes conflict graph $\mathcal{G}$ and bid vector corresponding to each operator $ \{ b_i \}_{\{ i \in N \}}$. Binary channel allocation vector $\{ X_i\}_{\{ i \in N \}}$ and the payment vector $\{p_i\}_{\{ i \in N \}}$ for all the operators are initialized to zero. Algorithm starts with determining the maximum bidding operator and its critical neighbor $C(i^*) = \argmax \limits_{k \neq i}\Lambda_k^{i^*}, ~k \in O_i*$. Channel allocation vector, $\{ X_i\}$ for the maximum bidding operator (winner) is updated to $1$ and the payment for the winning operator is updated to the price of the critical neighbor valuation, $\sigma_{i^*}^{c}$. The conflict graph $\mathcal{G}^{'}$ is updated with the remaining nodes after the removal of the nodes corresponding to the winning operator $i^*$ and its neighboring nodes $\mathcal{N}_{i^{*}}$. Repeat the process until $\mathcal{G}^{'}$ is NULL.

\par \textbf{ Example:} Consider a network of $3$ operators $A, B, C$, where each operator has $3$ base stations deployed in the region to provide services to the subscribers. The base stations $A_1, A_2, A_3$, $B_1, B_2, B_3$ and $C_1, C_2, C_3$ correspond to the operators $A$, $B$ and $C$ respectively. Conflict graph is illustrated in the Figure \ref{algo_1} based on the interference criteria.
 
 \begin{figure}
        \centering
        \begin{subfigure}[b]{\linewidth}
        \centering
                \includegraphics[width=0.7\textwidth]{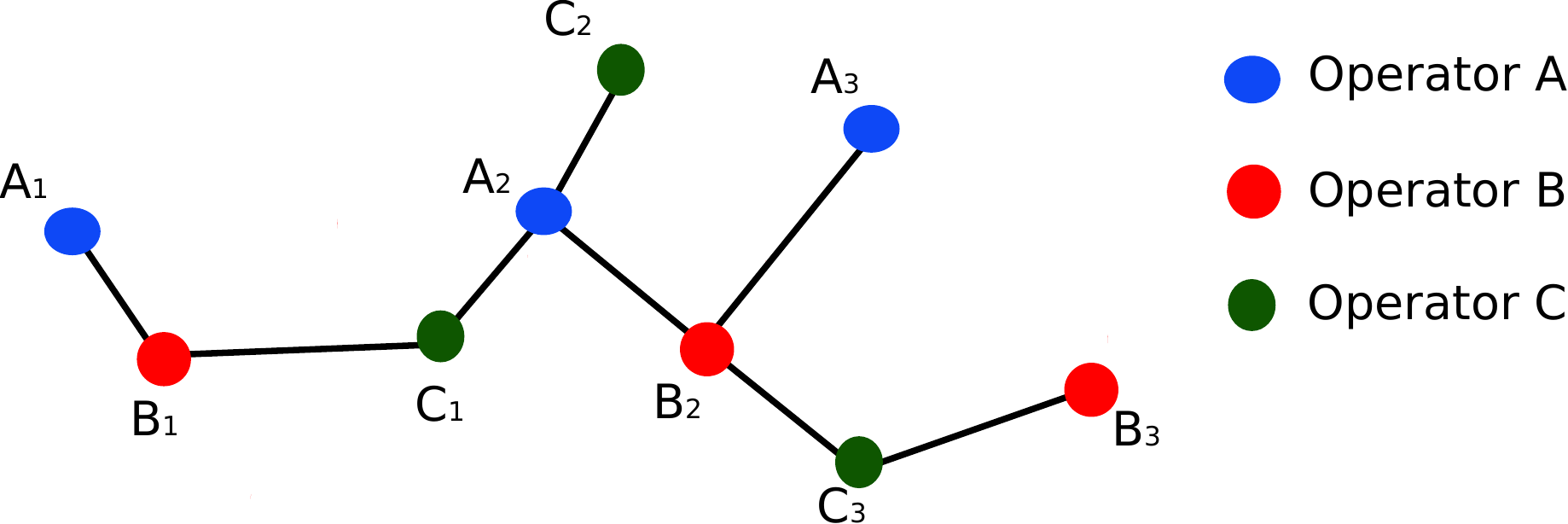}
                \caption{}
                 \label{algo_1}
        \end{subfigure}\\
        
      \begin{subfigure}[b]{\linewidth}
      \centering
                \includegraphics[width=0.6\textwidth]{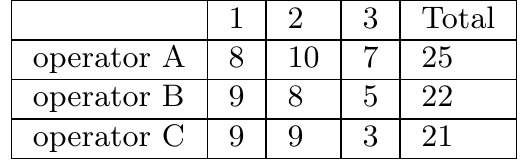}
                \caption{}
              \label{t:bid_vector}
        \end{subfigure}
        \caption{ Network of $3$ operators (a) Conflict Graph (b) Bid vector table corresponding to operator A, B and C.}
\end{figure}
 \par In the table of Figure \ref{t:bid_vector}, bid vector of each operator is presented. Operator $A$ has the highest bid value of $\sigma_A^b = 25$. Therefore, Operator $A$ will get channel at its base stations and it has to pay the price of its critical operator. As per the Definition (\ref{def: critical neighbor}), critical operator for winning operator $A$ is operator $C$ and $p_A = \sigma_A^c = 18$. Thus, the utility of operator $A$, $\mathcal{U}_A = 7$.  Update the conflict graph with the base stations of operator $B$ and $C$ not in conflict with base stations of operator $A$. 
 In second iteration, the updated $\mathcal{G}$ comprises of base stations $B_3$ and $C_3$. Operator $B$ wins the channel and pays the price, $ \sigma_B^c = 3$. The utility of operator $B$ is $2$. Operator $C$ does not get channel. 
 
 Now, if Operator $B$  tries to increase its utility by deviating from its true valuation $\sigma_B^v = 22$ to $\sigma_B^b = 28$ by increasing the bid of its base stations. Operator $B$ will get channel being the highest bidder among the operators. But, it has to pay the price of its critical operator which is operator $A$ and pays $ \sigma_B^c = 25$. This leads to negative utility $-3$, for operator $B$. Thus, bidding at true valuation is the best strategy for an operator in the auction. 
 Next, we prove that the proposed algorithm follows monotonicity and strategy-proofness.

 \begin{lemma}
 If any operator $i$ is allocated a channel by bidding at $\sigma_i^{b}$, it will also be allocated if it bids $\sigma_i^{b'}$, where $\sigma_i^{b'} \geq \sigma_i^{b}$ provided all the other operators bids remain unchanged.
\end{lemma}
\begin{proof}
 As stated in the algorithm, all the operator bids are arranged in non-increasing order of the bids $\sigma_i^{b}, \forall i \in N$. Let us assume in the sorted list $S$  operator $i$ lies at position $k$. Now, keeping all the other operator bids unchanged, increase the bid of operator $i$ to $\sigma_i^{b'}$, and again arrange all the operator bids in non-increasing order in another sorted list $S'$. Let us say, the position of operator $i$ in $S'$ is $l$, where $l \leq k$. Thus, the operator moves higher in the position which ensures it still receives the channel. This completes the proof.
\end{proof}

\begin{theorem}Algorithm \ref{channel_allocation} is strategy-proof.
\end{theorem}

\begin{proof}
To show the strategy-proofness of the algorithm, possible scenarios can be divided into two categories:\\
\textit{Scenario 1} : Any operator $i$ tries to deviate from truthfulness by bidding greater than the true valuation i.e, $\sigma_i^{b} > \sigma_i^{v}$.
\begin{enumerate}
 \item[] \textit{Case (i)}: An operator $i$ does not win channel even after bidding untruthfully at $\sigma_i^{b}$, greater than  $\sigma_i^{v}$. Hence, it will have utility, $\mathcal{U}_i = 0$.
 \item[] \textit{Case (ii)}: An operator $i$ wins channel at its bidding valuation, $\sigma_i^{b}$ (which is greater than the true valuation) as well as its true valuation, $\sigma_i^{v}$. It will have positive utility, $\mathcal{U}_i = \sigma_i^{v} - p_i$, which is same as in case operator bids at the true valuation. Thus, bidding at higher valuation does not lead to any extra incentive.
  \item[] \textit{Case (iii)}: An operator wins channel at $\sigma_i^{b}$, but looses at $\sigma_i^{v}$. Though the operator has been allocated channel by bidding untruthfully at higher value but it has to pay the price of its critical neighbor which is greater than the price of its true valuation (from the Algorithm \ref{channel_allocation}).
  
 \begin{equation*}
 \begin{split}
   \mathcal{U}_i &= \sigma_i^{v} - p_i,\\
   &= \sigma_i^{v} - \sigma_i^{c} \quad \text{where } p_i = \sigma_i^{c},\\
   &\leq 0. \quad \because \sigma_i^{v} < \sigma_i^{c}.
  \end{split}
 \end{equation*}
  
\end{enumerate}

\noindent \textit{Scenario 2} : An operator $i$ tries to deviate from truthfulness by bidding less than the true valuation i.e, $\sigma_i^{b} < \sigma_i^{v}$.
\begin{enumerate}
 \item[] \textit{Case (i)}: Operator $i$ looses channel at $\sigma_i^{b}$ as well as its true valuation, $\sigma_i^{v}$. Thus, it will have $\mathcal{U}_i = 0$.
 \item[] \textit{Case (ii)}: Operator $i$ wins channel at $\sigma_i^{b}$ as well as its true valuation, $\sigma_i^{v}$ which follows from monotonicity. Thus, it will have $\mathcal{U}_i = \sigma_i^{v} - p_i$.
  \item[] \textit{Case (iii)} : Operator looses at $\sigma_i^{b}$, but wins bidding at $\sigma_i^{v}$. Thus, the operator is at loss by deviating to untruthful value with zero utility. But, bidding truthfully at $\sigma_i^{v}$ result in the channel allocation, with utility $\mathcal{U}_i = \sigma_i^{v} - p_i$.
\end{enumerate}
From the above scenarios, it can be seen that bidding at $\sigma_i^{b} \neq \sigma_i^{v}$, does not improve the utility of operator. Thus, $\sigma_i^{b} = \sigma_i^{v}$ is the weakly dominant strategy for any operator. This completes the proof.
\end{proof}

\section{Complexity Analysis}
\label{sec:comp_complexity}
In this section, we study the computational complexity of the proposed algorithm. Computational complexity analysis of the proposed algorithm in case of $n$ operators and a single channel available for allocation among the base stations of the operators in the region with given conflict graph $\mathcal{G} = (V,\mathcal{E})$ is as follows.

The overall complexity of the proposed algorithm can be performed in two steps, channel allocation and price charging computation from the winning operator. The algorithm takes $\mathcal{O}(n)$ time to obtain the maximum bidding operator and allocates the channel among its base stations.

To calculate the price charged by the winning operator, algorithm needs to examine the conflicting neighbors corresponding to the base stations of winning operator. In any graph $\mathcal{G} = (V,\mathcal{E})$, complexity of determining the conflicting nodes is given by $\mathcal{O}( |V| + |\mathcal{E}|)$, where $V$ is the number of vertices (nodes) and $\mathcal{E}$ corresponds to number of edges in the graph. The edges in a graph with $V$ nodes is bounded as $ \mathcal{E} \leq \frac{V(V-1)}{2}$. Let us assume $i^{th}$ operator has $m_i$ base stations and the total number of base stations is obtained as $m = \sum_{i \in N} m_i$. Since, the base stations of the same operator are not in conflict with each other, the winning operator $i$ can conflict with maximum of remaining $ m - m_i$ base stations of the given conflict graph. Therefore, the  complexity of determining critical neighbor is $\mathcal{O}\big(m_i + \frac{(m - m_i)(m - m_i-1)}{2}\big)$. The term $ m - m_i$ in the complexity of critical neighbor will keep on decreasing over the iterations, as the number of base stations allocated channel would increase. Thus, the complexity of critical neighbor evaluation is $\mathcal{O}\big(m_i + (m - m_i)^2 \big)$. On simplification, this is approximately $\mathcal{O}( m^{2})$. The complexity of algorithm for an iteration (channel allocation and price charging) is $\mathcal{O}(n + m^2)$. Therefore, the overall complexity of the algorithm is $\mathcal{O}\big(n(n + m^2)\big)$, where $n$ is the number of operators and $m$ is the total number of base stations in the given region. As the number of base stations is much larger than the number of operators, the overall complexity of the algorithm can be approximated as $\mathcal{O}(nm^2)$.

 \par The computational complexity of (VCG) \cite{vickrey1961counterspeculation,clarke1971multipart,groves1973incentives}  mechanism for optimal channel allocation is given by $\mathcal{O}(2^m -1)$, where $m$ denotes the total number of base stations across all the operators. Similarly, the computational complexity of the price charging mechanism is $\mathcal{O}(m(2^m -1))$. However, the proposed algorithm has overall complexity $\mathcal{O}(nm^2)$, which makes it efficient and feasible even for large set of base stations of the operators.
\section{Simulation Results}
\label{sec:sim_results}
In this section, we compare the performance of the proposed algorithm with that of the VCG mechanism. The simulations are performed in MATLAB \cite{moler1982matlab}. As discussed in Section \ref{sec:sys_model} the system model has multiple operators in the region bidding for its base stations. In Figure \ref{perform_comp}, simulations results for three operator scenario is compared with those of the classical VCG algorithm for the following parameters:
\begin{itemize}
\item Allocation Efficiency: It is defined as the total number of base stations which are assigned channels across all the operators i.e., $\sum \limits_{i \in N} x_{i}$, where $x_i$ denotes the allocation vector of operator $i$.
\item Social Welfare: Social welfare is defined as the sum of the bids corresponding to the base stations which are allocated channels. 
\end{itemize}
 
\par We consider LTE base stations of all the operators are uniformly distributed across the region with one channel availability in spectrum database. From the distribution, $\mathcal{G}$ is generated. VCG mechanism computes the sets of feasible outcomes, $\Omega$ which comprises of all the independent set from the graph. Based on the bids, the set of all the weighted independent feasible outcomes is obtained and then the the maximum weighted independent set is chosen. 

 It has been observed from the Figure \ref{perform_comp} that the allocation efficiency
(number of base stations that are allocated channels) and the social welfare is near to the optimal value obtained from the VCG mechanism. The reduced complexity of the proposed algorithm makes it suitable for the spectrum allocation among the large set of base stations of multiple operators. Thus, auction can be performed in short durations based on the dynamic load variation in the network.
\begin{figure}
        \centering
        \begin{subfigure}[h]{\linewidth}
        \centering
                
                \includegraphics[width=0.65\textwidth]{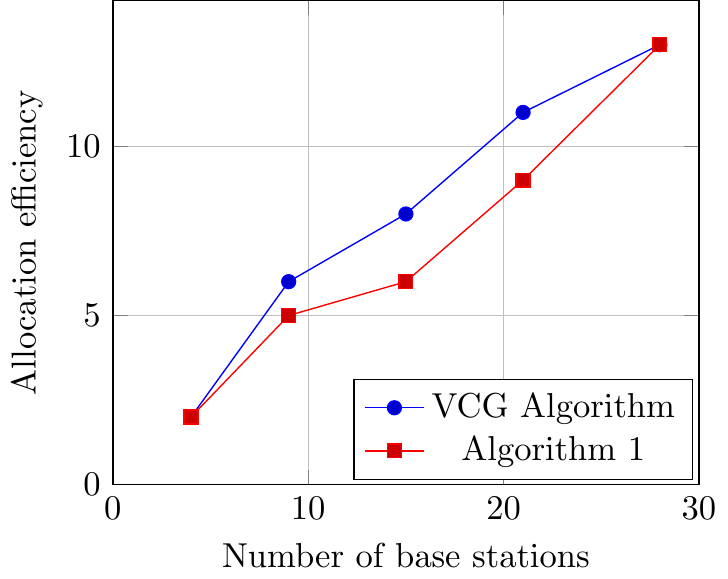}
                \caption{}
                \label{alloc_eff}
        \end{subfigure}\\
        
      \begin{subfigure}[h]{\linewidth}
      \centering
               
                \includegraphics[width=0.65\textwidth]{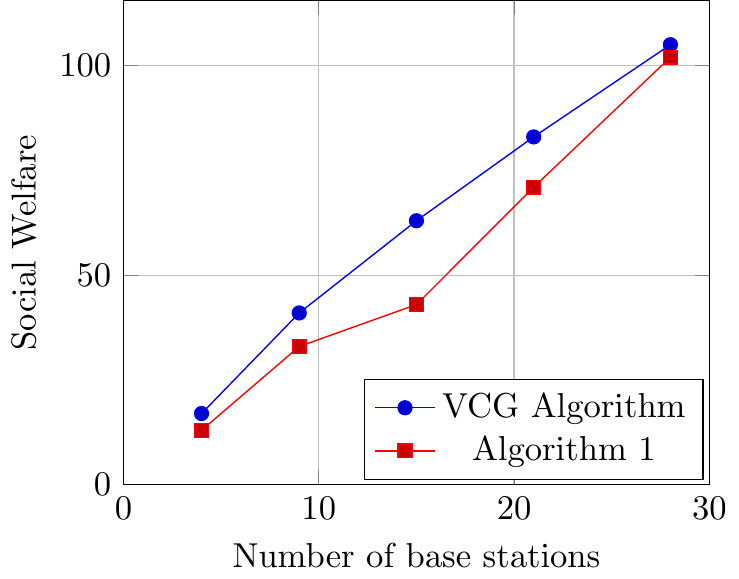}
                \caption{}
              \label{soc_wel}
        \end{subfigure}
\caption{ Performance comparison of the VCG and Algorithm $1$ in three operator scenario (a) Allocation efficiency (b) Social welfare.}
\label{perform_comp}
\end{figure}

\section{Conclusions}
\label{sec:conclude}
In this work, we study the problem of spectrum allocation in multi-parameter environment and propose an algorithm for spectrum allocation among the multiple base stations of multiple operators. We prove that the proposed algorithm guarantees strategy-proofness. We also present computational complexity of the proposed algorithm and prove that the proposed algorithm can be implemented in polynomial time $\mathcal{O}(nm^2)$, where $n$ is number of operators and $m$ represents the total number of base stations across all the operators. Thus, the proposed algorithm is feasible for practical implementation in real time spectrum auctions. Using simulation results, we observe that the performance of the proposed algorithm is near-optimal in terms of allocation efficiency and social welfare, and it solves the intractability issue of VCG mechanism. 

\bibliographystyle{IEEEtran}
\bibliography{refer}


\end{document}